\documentclass[11pt,onecolumn]{article}
\usepackage{fullpage}
\usepackage{times}
\usepackage{titlesec}
\usepackage{appendix}
\usepackage{cases}
\usepackage{amsthm}
\usepackage{epstopdf}
\usepackage{mathrsfs}
\usepackage{colortbl}
\usepackage{url}
\usepackage{threeparttable}
\usepackage{multirow}
\usepackage{subfigure}
\usepackage{epsfig}
\usepackage{algorithm,algorithmic}
\usepackage{caption2}
\newtheorem{theorem}{Theorem}

\newtheorem{corollary}{Corollary}
\newcommand{\para}[1]{\vspace{0.01in}\noindent\textbf{#1 }}
\newcommand{\sjqxzlhy}[1]{\textcolor{black}{#1}}

\newcommand{\hstanq}[1]{\textcolor{black}{#1}}
\newcommand{\hstan}[1]{\textcolor{black}{#1}}

\begin{document}
\title{\bf Optimal Rendezvous Strategies for Different Environments in Cognitive Radio Networks}
\author{Haisheng Tan$^1$\thanks{Corresponding Author. Email: thstan@jnu.edu.cn.}  \quad Jiajun Yu$^2$\quad Hongyu Liang$^2$ \quad Tiancheng Lou$^3$\quad Francis C.M. Lau $^4$ \\\\
$^1$ Department of Computer Science, Jinan University, Guangzhou, China\\
$^2$ Institute for
Theoretical Computer Science, IIIS, Tsinghua University, Beijing,
China\\
$^3$ Google Inc., California, USA\\
$^4$ Department of Computer Science, The
University of Hong Kong, Pokfulam, Hong Kong\\}

\date{}
\maketitle

\begin{abstract}

In Cognitive Radio Networks (CRNs), the secondary users (SUs) are allowed
to access the licensed channels opportunistically. A fundamental and
essential operation for SUs is to establish communication through
choosing a common channel at the same time slot, which is referred to as
rendezvous problem. In this paper, we study strategies to achieve
fast rendezvous for two secondary users.

The channel availability for secondary nodes is subject to temporal
and spatial variation. Moreover, in a distributed system, one user is
oblivious of the other user's channel status. Therefore, a fast rendezvous
is not trivial. Recently, a number of rendezvous strategies have been
proposed for different system settings, but rarely have they taken
the temporal variation of the channels into account.
In this work, we first derive a time-adaptive
strategy with optimal expected time-to-rendezvous (TTR) for
synchronous systems in stable environments, where channel
availability is assumed to be static over time. 
Next, in dynamic environments, which better represent 
temporally dynamic channel availability in CRNs, we first derive optimal
strategies for two special cases, and then prove that our strategy is
still asymptotically optimal in general dynamic cases.

Numerous simulations are conducted to demonstrate the performance of our
strategies, and validate the theoretical analysis. The impacts of
different parameters on the TTR are also investigated,
such as the number of channels, the channel open possibilities, the
extent of the environment being dynamic, and the existence of an intruder.
\end{abstract}
\newpage
\section{Introduction}\label{sec:intro}
\subsection{Background}
The number of wireless devices has skyrocketed over the last decade,
which exasperates the scarcity of spectral resources. The reality is that
most wireless spectrum bands have been allocated by regulatory agencies
to licensed users, which however are severely under-utilized. It is
reported that the utilization of licensed spectrum varies from
$15\%$ to $85\%$~\cite{cnAkyildiz}, and only $5\%$ of the spectrum
from $30$ Mhz to $30$ GHz is used in the US~\cite{sigcomm09Bahl}.
Recently some regulators have issued permissions for unlicensed
devices to use parts of the licensed spectrum under restrictions. %
Cognitive Radio (CR) realizes the unlicensed
devices (called \emph{secondary users}) to utilize the temporarily
unused licensed spectrums without interfering with the licensed
devices (called \emph{primary users}). Therefore, CR is a promising technology to alleviate the spectrum
shortage problem in wireless communication, and Cognitive Radio Network
(CRN) is considered the next generation of communication networks.

In a CRN, through spectrum sensing~\cite{ss11Taricco}, each secondary
user (SU) has the ability to detect current open (or available)
channels at its site, i.e., channels that are not occupied by the
primary users (PUs). Due to dynamic come and go of PUs, the
available channels to SUs have the following important
characteristics: 1) \emph{Spatial
Variation}: SUs at different locations may have
different available channels; and 2)\emph{Temporal Variation}: the
available channels of a SU may change over time.

Spectrum assignment is to allocate available channels to SUs to
improve network performance, such as connectivity, spectrum
utilization, network
throughput and fairness~\cite{Infocom12Li,Infocom12Lu,JsacRen,mobihoc12Huang,mobihoc07Yuan}.
Spectrum assignment is one of the most challenging problems in CRNs.
One fundamental action in spectrum assignment is referred to as
the \emph{rendezvous} problem of which the simplest form can be stated as:

 \emph{For two secondary users, Alice and Bob, how do
they establish a connection through choosing a common channel at the same
time slot?}

We call the approach they adopt for choosing a
channel at each round as \emph{a rendezvous strategy}. A strategy attaining
a fast time-to-rendezvous (TTR) is not trivial due to the following facts: 1) the local available channels for Alice and Bob are asymmetric; 2) in a distributed system, one SU only knows his/her local channel status but not the
other's; 3) the local channel availability may change during the rendezvous; and 4) in asynchronous systems, rendezvous is even more
complicated as one does not know the time of the other SU or how many
tries the other SU has made.

\subsection{Related Work}
\subsubsection{Basic Rendezvous Problem and its Variants}
In the basic form of the Rendezvous problem (also called the
telephone coordination game \cite{TelCoGame}), two players are
placed in two separated rooms each with $n$ telephone lines. There
is a one-to-one matching between the lines in two rooms,
and players can only use matched lines to contact each other;
the matching is not known to either of the players. In every
round each player can select a line in his room and see whether it
is connected with the line chosen by another player. The goal is to
design a strategy for the players to get connected while minimizing
the number of trial rounds.

It is not hard to see that the optimal strategy takes $n/2$ rounds
in expectation, which is achievable by letting the first player
choose a random line and keep using it, and letting the second
player try the lines one by one in a random order. However, this
strategy only works when the players can correctly determine who is
the first player, which may not be possible in a given application.
Thus, a player-independent strategy is desired. Anderson and Weber \cite{jap90} presented such
a strategy using at most $0.829n$ rounds in expectation, where each player
will repeat either to keep on choosing one line or randomly choosing a line
different from his last choice with some
probabilities. There are many other well-studied
variants of the Rendezvous problem; see, e.g.,
\cite{jap99,AlpernGal03} and the references therein. Also, some researcher considered jammers in rendezvous, such as the node discovery in (conventional) multiple-channel wireless communications~\cite{DCOSS09Roger}.

\vspace{-2mm}\subsubsection{Rendezvous in CRNs}
In contrast to the basic version, rendezvous in CRNs happens in an asymmetric
dynamic scenario, where each user may have different and time-variant sets
of available channels. There have been a number of rendezvous strategies
proposed for CRNs. One group of these strategies makes use of a dedicated
channel, called a Common Control Channel (CCC), to exchange information
between users for rendezvous (e.g.,
~\cite{JcommC,JsacJia,dyspan07PR,dyspan05Zhao}). However, assuming a
CCC might not be practical as the channel could be occupied by some PUs,
and it also introduces an easy
attack point. \emph{Blind rendezvous} without
any centralized controller or a dedicated CCC is therefore
preferred (in this paper, unless
specified otherwise, rendezvous stands for this blind case). A main group of
these rendezvous strategies adopts the channel-hopping (CH) technique, where
SUs hop among their available channels based on a hopping sequence to
achieve rendezvous (e.g.,~\cite{SECON12Gandhi,Infocom11Lin,TMCTheis}).
Quorum Systems are also frequently adopted for rendezvous in CRNs(e.g.,~\cite{ADHOC13}). Besides, \hstan{in paper~~\cite{esaAzar}, Azar et. al. proposed an
 strategy using geometric distribution for asynchronous systems, and proved it to be optimal when there are a large number of channels.} Essential details of their strategy
will be covered in
Section~\ref{sec:stable:azar}. Although most of the above work have
considered the spatial variation of the channels, i.e., users in different
locations may have different sets of available channels, they are limited
to stable (channel) environments where the channel availabilities are static
over time. Not many of them take the temporal variation of channels into
account. In this work, we will study both the stable and the dynamic
environment in asynchronous and synchronous systems. We hope this work may
inspire further research on the dynamic version of the rendezvous problem.

\vspace{-3mm}\subsection{Our Results}

In this paper, we investigate the rendezvous problem for two SUs and derive
optimal strategies for different system settings. For the sake of
theoretical analysis, 
we assume each channel of Alice and Bob has an available probability of $p_a$
and $p_b$ respectively, and denote the total number of channels as $n$. Our results can be summarized as follows:
\begin{itemize}
\vspace{-2mm}\item For a synchronous stable environment, we derive a 
time-adaptive strategy that guarantees successful rendezvous at the first
common channel, say channel $m$, within at most $m$ rounds. The expected
TTR of our strategy is $\frac{2-\max(p_a,p_b)}{\min(p_a,p_b)}$ when
$n\rightarrow +\infty$, which is a
2-approximation to the optimal.
\vspace{-2mm}\item Our main effort is devoted to the dynamic environments which better
reflect the nature of
temporal variation of channel availability in CRNs. 
We first define two special cases, the semi-stable and the independent
dynamic cases (refer to Section~\ref{sec:dynamic}). For the
synchronous semi-stable case, we derive an optimal strategy based on the
one used in the synchronous stable environment. For the independent dynamic
environment, we derive a simple stationary strategy and prove that it is
exactly optimal no matter whether there is a common clock or not.  When
$n\rightarrow +\infty$, its expected TTR is
$\frac{1}{p_a}+\frac{1}{p_b}-1$. Then, we model the channel availability in
the general dynamic environment as a Markov process. When neither of the
environments of two SUs is stable nor semi-stable, we prove the
expected TTR of our strategy is
$O\left(\ln\left(\frac{1}{\min(p_a,p_b)}\right)\cdot
\left(\frac{1}{p_a}+\frac{1}{p_b}-1\right)\right)$, which is asymptotically optimal.
\vspace{-2mm}\item Based on 
simulation, we validate the above theoretical
analysis and demonstrate the efficiency of our strategies when there are a
small number of channels. Besides, we reveal the impacts of different
parameters on the TTR, such as the number of channels, the channel open
possibilities, the extent to which the environment being dynamic,
and the existence of an intruder.
\end{itemize}
\para{Paper Organization:}
In Section~\ref{sec:def}, we formally define our model and problems
studied in this paper. We investigate the rendezvous problem in
stable environments in Section~\ref{sec:stable}. The different cases
in dynamic environments are studied in Sections~\ref{sec:dynamic}.
Section~\ref{sec:sim} gives the simulation results and discussions.
The whole paper is concluded in Section~\ref{sec:con} with possible
future works.

\vspace{-3mm}\section{Problem Definition}\label{sec:def}

We consider a pair of secondary nodes, called Alice and Bob, which
need to establish a connection. There are totally $n$ channels with
ID's from $1$ to $n$. A binary vector $A^t=\{A_1^t, A_2^t,\ldots,
A_n^t\}$ indicates whether a channel is open for Alice at time $t$, e.g., if channel $i$ is open
\sjqxzlhy{at time $t$}, $A_i^t=1$; otherwise, $A_i^t=0$. Vector $B^t$ is similarly defined for Bob.
Through spectrum sensing, both Alice and Bob can know their local
available channels. However, as in a distributed system, Alice and Bob
are oblivious of $B^t$ and $A^t$, respectively. In addition, there is no information (i.e., the node IDs) to break the symmetry of the two nodes.

We investigate both \emph{synchronous} and \emph{asynchronous}
distributed systems in \emph{stable} and \emph{dynamic}
environments. In a synchronous system, Alice and Bob will have a
common clock, \sjqxzlhy{whereas in an asynchronous system, they do not
know each other's clock.} As mentioned before, in a stable
environment, the channel availability will be static over time,
while in a dynamic environment, the channel status may change over
time. We assume at least one common channel exists in the channel
environment, or otherwise the rendezvous can never happen\footnote{In dynamic
environments, there could be no common channel at some time slots.}. Within
a round (time slot), each node will try to achieve rendezvous once using
a strategy which is stationary or adaptive over the rounds.
\begin{table}
\centering
\caption{Definitions of the symbols}\label{table:symbol}
\begin{tabular}{p{0.16\textwidth}p{0.72\textwidth}}
\hline
 \multicolumn{1}{p{0.16\textwidth}}{Symbol} &  \multicolumn{1}{p{0.68\textwidth}}{Definition} \\
 \hline
$n$ & the total number of channels\\
 $A_i^t$ or $B_i^t$   &  availability of channel $i$ for Alice or Bob at time $t$\\
  $\mathcal{S}_a$ or $\mathcal{S}_b$ & the strategy of Alice or Bob \\
 $\mathcal{S}_a^{t}(i)$  or $\mathcal{S}_b^{t}(i)$ &  \sjqxzlhy{the probability} that Alice or Bob chooses channel
 $i$ at time slot $t$\\
 $\phi_t$ & the flag \sjqxzlhy{indicating} whether a rendezvous is achieved successfully
 at time $t$\\
 $TTR$ & the time-to-rendezvous\\
 \hline
\end{tabular}
\end{table}

In CRNs, one node might wake up and start trying rendezvous
earlier than the other. These
tries will definitely fail. Thus, we count the time-to-rendezvous
(TTR) from the starting point when both Alice and Bob have
waken up. Based on the symbols in Table~\ref{table:symbol}, the
possibility for a successful rendezvous at $t$ is $Pr[\phi_t=1 | \{A_i^t\},\{B_i^t\}]=\sum^{n}_{i=1} A_i^tB_i^t\mathcal
{S}_a^t(i)\mathcal {S}_b^t(i).$
 The TTR is the first instant
when Alice and Bob choose a common channel simultaneously: $TTR =\min\{t: \phi_t=1\}$.
\sjqxzlhy{Our} goal is to derive 
strategies for Alice and Bob that minimize $E[TTR]$, the
expectation of TTR.

For convenience of analysis, at any round we assume each channel
$i$ of Alice has a probability of $p_a $ to be open
($0<p_a\leq 1$) independently. That is to say, for Alice, in the dynamic environment, without the knowledge of the channel status in previous slots, $Pr[A^t_i=1]=p_a$, $ \forall t\geq 1$; in the stable environment, each channel will open with probability of $p_a$ at the first time slot and never change status subsequently. $p_b$ is similarly defined for Bob.

\vspace{-2mm}\section{Stable Environment}\label{sec:stable}
\subsection{Strategy for Asynchronous Systems}\label{sec:stable:azar}
In an asynchronous system, as there is no common clock, neither of
Alice and Bob know the time of the other player \sjqxzlhy{nor} how many
time slots the other player has tried for rendezvous. Therefore, a
strategy that is adaptive to time is meaningless, and stationary
strategies are \sjqxzlhy{required}. In \cite{esaAzar}, the authors
proposed a stationary strategy based on
geometric distributions shown as Strategy A in Appendix~\ref{app:azar}.
Its expected TTR is proved to be $O(\frac{1}{p_ap_b})$ when
$n\rightarrow +\infty$. In addition, it is proved that the strategy
is essentially optimal as the TTR of any stationary strategy is
$\Omega(\frac{1}{p_ap_b})$. However, a main weakness of Strategy A is that
Alice (Bob) has to have the knowledge of $p_b$ ($p_a$)\footnote{The authors
also considered a third
party, called Eve, which is treated as an intruder in
Section~\ref{sec:sim:impact:q}.}, which might be not feasible in practice. 

Next, we will extend their work and present our simple optimal strategy in
synchronous stable environments which guarantees a fast rendezvous and is
applicable for finite channels.

\vspace{-2mm}\subsection{Strategy for Synchronous Systems} \label{sec:stable:syn}
In a synchronous system, although one user is oblivious of the other user's
channel status, he/she is aware of the time of the other user. Therefore, we
can derive non-stationary strategies
that \sjqxzlhy{are} adaptive to time:
\begin{table}[htpb]\centering
\begin{tabular}{p{0.95\textwidth}}
 \hline
Strategy B:\\
\hline \sjqxzlhy{In} the $i$-th round ($i\geq 1$),\\
 \quad $\mathcal {S}_a$:  \sjqxzlhy{Alice chooses her first local open channel from channel $i$ to channel $n$; that is, she chooses channel $a^*=\min\{a~|~A_a^{i}=1;i\leq a\leq n\}$.}\\
\quad $\mathcal {S}_b$:  \sjqxzlhy{Bob chooses his first local open channel from channel $i$ to channel $n$; that is, he chooses channel $b^*=\min\{b~|~B_b^{i}=1; i\leq b\leq n\}$.}\\
\hline
\end{tabular}
\end{table}
Strategy B is actually \sjqxzlhy{a} novel \emph{waiting-to-meeting} scheme: the one who reaches the first
common open channel will stick to it until the other reaches that same
channel when they achieve rendezvous. Therefore, we have this theorem:
\begin{theorem}
Suppose the first common open channel of Alice and Bob is channel
$m$. With Strategy B, they will definitely achieve rendezvous on
$m$ with no more than $m$ rounds in synchronous stable environments.
\end{theorem}

Next, we prove the optimality of our strategy.
\begin{theorem}\label{thm:perfB}
\sjqxzlhy{In the synchronous stable environment, when
$n\rightarrow +\infty$, the expected TTR of Strategy B
satisfies $E[TTR]\leq \frac{2-\max(p_a,p_b)}{\min(p_a,p_b)},$
which is a 2-approximation to the optimal.}
\end{theorem}

\begin{proof}
At $i$-th round, if Alice chooses channel $i+k$ \sjqxzlhy{where
$0\leq k \leq n-i$}, it means \sjqxzlhy{that channels
$i,i+1,\ldots,i+k-1$ are all closed and that channel $i+k$ is open}.
In addition, as in a stable environment, the channel status at the
$i$-th round is the same as that in the first round. Therefore, the
\sjqxzlhy{probability} that Alice chooses channel $i+k$ on the
$i$-th round is $P^a_{i,k}=(1-p_a)^kp_a$.
Similarly, Bob chooses channel $i+k$ with \sjqxzlhy{probability} $P^b_{i,k}=(1-p_b)^kp_b$.
Thus, on the $i$-th round, the rendezvous \sjqxzlhy{probability} is
\sjqxzlhy {
\vspace{-2mm}\begin{eqnarray}
Pr[\phi_i=1] &=& \sum_{k=0}^{n-i}P^a_{i,k}P^b_{i,k}=p_ap_b\sum_{k=0}^{n-i}(1-p_a)^k(1-p_b)^k \nonumber \\
&=& \hstan{p_ap_b
\cdot\frac{1-((1-p_a)(1-p_b))^{n-i+1}}{1-(1-p_a)(1-p_b)} \label{eqn:phii1}}\\
&=& \hstan{\frac{p_ap_b}{p_a+p_b-p_ap_b}, \quad \textrm{when~} n\rightarrow
+\infty.} \label{eqn:phii2}
\end{eqnarray}
} According to Eqn~\sjqxzlhy{(\ref{eqn:phii2})}, when $n\rightarrow
+\infty$, the rendezvous possibility of each round converges to a
constant\footnote{In fact, we can see from
\sjqxzlhy{Eqn~(\ref{eqn:phii1}) that}, even when $n$ is a
\sjqxzlhy{small finite integer}, e.g. $n=20$, $\phi_i$ is
\sjqxzlhy{still close to the derived constant}. Simulations will
demonstrate the efficiency of our strategy when there are finite
channels.}. Without loss of generality, we set $p_a\leq p_b$. We can
\sjqxzlhy{estimate} the expected TTR as $E[TTR]= \frac{p_a+p_b-p_ap_b}{p_ap_b}=\frac{p_a(1-p_b)+p_b}{p_ap_b}
\leq\frac{p_b(1-p_b)+p_b}{p_ap_b}=\frac{2-p_b}{p_a}<\frac{1}{\min(p_a,p_b)}\times 2$.
 Further, as mentioned in
Appendix~\ref{app:azar}, \hstan{a trivial lower bound of $E(TTR)$ is $\frac{1}{\min(p_a,p_b)}$}. Thus, Strategy B is
a 2-approximation to the optimal.
\end{proof}



\vspace{-6mm}\section{Dynamic Environment}\label{sec:dynamic}
As mentioned before, the channel availability
for secondary users actually is dynamic in time. In this section, we
study the rendezvous strategies for SUs in dynamic environments by starting
with two special cases and then investigating the general cases.

\vspace{-2mm}\subsection{Special Case 1: semi-stable
environment}\label{sec:dynmaic:semi-stable}


In the first special case, for Alice and Bob, once a channel is
open (closed) at one time slot, it will definitely
change its status to close (open) at the next time slot. We call this case the \emph{semi-stable environment}.
One common property between semi-stable and stable
environments is that, once the status of a channel at a
time slot is known, we can correctly compute its status at any other time.

In a semi-stable environment, if we only consider the odd or even rounds,
it is equivalent to a stable environment. Therefore, if there is no common
clock, Strategy A using geometric distribution is still essentially
optimal, since its expected TTR is at most twice the time
in \sjqxzlhy{a} stable environment, which is $O(\frac{1}{p_ap_b})$. Similarly, when there is a common clock, we can modify Strategy
B a bit and get Strategy $\widetilde{B}$, as follows:
\begin{table}[htpb]\centering
\begin{tabular}{p{0.95\textwidth}}
 \hline
Strategy $\widetilde{B}$:\\
\hline \sjqxzlhy{In} the $i$-th round ($i\geq 1$),\\
 \quad $\mathcal {S}_a$:  \sjqxzlhy{Alice chooses her first local open channel from channel $\lceil \frac{i}{2}\rceil$ to $n$; that is, she chooses channel $a^*=\min\{a~|~A_a^{i}=1;\lceil \frac{i}{2}\rceil\leq a\leq n\}$.}\\
\quad $\mathcal {S}_b$:  \sjqxzlhy{Bob chooses his first local open
channel from channel $\lceil \frac{i}{2}\rceil$ to $n$; that is, he
chooses channel $b^*=\min\{b~|~B_b^{i}=1;\lceil
\frac{i}{2}\rceil\leq b\leq n\}$.}
\\
\hline
\end{tabular}
\end{table}

\hstan{Strategy $\widetilde{B}$ is still based on the waiting-to-meeting scheme. Its expected TTR
is at most twice the time of Strategy B in a stable environment.} Therefore, \sjqxzlhy{ the following corollary is straightforward.}

\begin{corollary}
\hstan{In a semi-stable environment, when $n\rightarrow +\infty$, 1) Strategy A  achieves essentially optimal expected TTR in asynchronous systems; and 2) Strategy $\widetilde{B}$ is 4-approximation in synchronous systems.}
\end{corollary}
\vspace{-2mm}\subsection{Special Case 2: independent dynamic
environment}\label{sec:dynamic:independent}

We come to another extreme special case, where for Alice and Bob, the event that a
channel is open at time $t$ is independent from the status of the same
channel at time $t'$ for any $t'<t$. Thus, we call this case \sjqxzlhy{the}
\emph{independent dynamic environment}. Recall that we assume at each round
a channel of Alice (Bob) has a probability of $p_a $ ($p_b$) to be open.

In \sjqxzlhy{an} independent dynamic environment, we give the following simple stationary \sjqxzlhy{strategy}, called Strategy C: at each round, Alice chooses her first local open channel, and Bob does similarly.
Note that Strategy C does not need a common clock. In addition, Strategy C
can not be applicable to stable environments, since Alice and Bob can never
achieve rendezvous successfully when the first open channels of them are
not common, even if there are other common channels. Similarly, it
can not be used in a semi-stable case.

For channel $i$, Alice will select it only if it is open and all the
channels before it are closed. At any time slot, for Alice, all her
channel status will be refreshed with open possibility of $p_a$.
Therefore, at any time slot, Alice select channel $i$ with
probability of $p_a(1-p_a)^{i-1}$. A similar conclusion can be
achieved for Bob. Further, at any round, the rendezvous probability
of Alice and Bob on channel $i$ is
$p_ap_b(1-p_a)^{i-1}(1-p_b)^{i-1}$. Thus, the following theorem
indicating the performance of Strategy C can be easily obtained.
\begin{theorem}\label{thm:perfC}
At any round \sjqxzlhy{$t$}, the rendezvous probability of Strategy C
is
\vspace{-2mm}\begin{equation}\label{eqn:14}
\hstan{Pr[\phi_t=1]=\sum^n_{i=1}p_ap_b(1-p_a)^{i-1}(1-p_b)^{i-1}\sjqxzlhy{,}}
\end{equation}\vspace{-2mm}
\sjqxzlhy{which is actually independent of $t$.} When $n\rightarrow
+\infty$, the expected TTR is $E[TTR]=\frac{1}{Pr[\phi_t=1]}=\frac{1}{p_a}+\frac{1}{p_b}-1.$
\end{theorem}

Further, the following theorem infers that our strategy is exactly
the optimal one.
\begin{theorem}\label{thm:upperbound}
In \sjqxzlhy{an} independent dynamic case, no matter
\sjqxzlhy{whether} there is a common clock or not, the rendezvous
of any strategy at any time slot $t$ satisfies $Pr[\phi_t=1] \leq \sum^n_{i=1}p_ap_b(1-p_a)^{i-1}(1-p_b)^{i-1}$.
\end{theorem}

\begin{proof}
For any strategy, it is reasonable
\sjqxzlhy{without loss of generality to assume} that Alice and Bob
will always make a try by choosing an open channel per round. At
time $t$, the \sjqxzlhy{probability} that Alice and Bob chooses
channel $i$ is denoted as $\mathcal{S}_a^{t}(i)$ and
$\mathcal{S}_b^{t}(i)$ respectively. We have at any time $t$, $\sum^n_{i=1}\mathcal{S}_a^{t}(i)=
\sum^n_{i=1}\mathcal{S}_b^{t}(i)=1.$
The \sjqxzlhy{probability} that they can achieve rendezvous at time
$t$ is $Pr[\phi_t=1]=\sum^n_{i=1}\mathcal{S}_a^{t}(i)\mathcal{S}_b^{t}(i).$
Set $\{i_1,i_2,\ldots,i_n\}$ as a permutation of numbers from $1$ to
$n$ such that $\mathcal{S}_a^{t}(i_1)\geq \mathcal{S}_a^{t}(i_2)\geq\cdots\geq\mathcal{S}_a^{t}(i_n).$
Similarly, $\{j_1,j_2,\ldots,j_n\}$ is set such that $\mathcal{S}_b^{t}(j_1)\geq
\mathcal{S}_b^{t}(j_2)\geq\cdots\geq\mathcal{S}_b^{t}(j_n).$
\sjqxzlhy{By Abel's Inequality (or simple mathematical
manipulations)}, we have $Pr[\phi_t=1]=\sum^n_{i=1}\mathcal{S}_a^{t}(i)\mathcal{S}_b^{t}(i)
\leq  \sum^n_{k=1}\mathcal{S}_a^{t}(i_k)\mathcal{S}_b^{t}(j_k).$
At each round, each channel of Alice has an open probability of
$p_a$, so $\mathcal{S}_a^{t}(i_1)\leq p_a$. Set
$\mathcal{S}_a^{t}(i_1)=p_a-\epsilon$ $(\epsilon\geq 0)$, and
$\sum^{n}_{k=2}\epsilon_k=\epsilon$ $(\forall 2\leq k \leq n,
\epsilon_k \geq 0)$. Then, we have
\vspace{-2mm}\begin{eqnarray}\label{eqn:epsilon}
Pr[\phi_t=1] & \leq &
\sum^n_{k=1}\mathcal{S}_a^{t}(i_k)\mathcal{S}_b^{t}(j_k)=(p_a-\epsilon)\mathcal{S}_b^{t}(j_1)+\sum^n_{k=2}\mathcal{S}_a^{t}(i_k)\mathcal{S}_b^{t}(j_k) \nonumber \\
&=&(p_a-\sum^{n}_{k=2}\epsilon_k)\mathcal{S}_b^{t}(j_1)+\sum^n_{k=2}\mathcal{S}_a^{t}(i_k)\mathcal{S}_b^{t}(j_k) \nonumber \\
&\leq & \hstan{p_a\mathcal{S}_b^{t}(j_1)+\sum^n_{k=2}(\mathcal{S}_a^{t}(i_k)-\epsilon_k)\mathcal{S}_b^{t}(j_k)}.
\end{eqnarray}
The last step in \sjqxzlhy{Eqn~(\ref{eqn:epsilon})} is due to
$\mathcal{S}_b^{t}(j_1)\geq \mathcal{S}_b^{t}(j_k)$ $(k\geq 2)$.

According to \sjqxzlhy{Eqn~(\ref{eqn:epsilon})}, in order to
maximize the rendezvous probability for a round, we should set
$\epsilon_k=0$ ($k\geq 2$). So we have $\epsilon=0$ and
$\mathcal{S}_a^{t}(i_1)=p_a$. \sjqxzlhy{Hence}, to get the largest
rendezvous probability, we have to set $\mathcal{S}_a^{t}(i_1)$
maximal. Due to a similar argument, for Alice,
$\mathcal{S}_a^{t}(i_2)$ should be maximized on the premise that
$\mathcal{S}_a^{t}(i_1)$ is set \sjqxzlhy{to be} maximal. That is to
say, once channel $i_1$ is open, we choose channel $i_1$; otherwise,
we choose $i_2$ as long as it is open. Inductively, to achieve the
largest rendezvous probability, all $\mathcal{S}_a^{t}(i_k)$ $(k\geq
2)$ must be maximized with the premise that
$\mathcal{S}_a^{t}(i_1,),\mathcal{S}_a^{t}(i_2),...,\mathcal{S}_a^{t}(i_{k-1})$
are set \sjqxzlhy{to be} maximal. Therefore, we \sjqxzlhy{can} set $\mathcal{S}_a^{t}(i_k)=p_a(1-p_a)^{k-1}$ when $k\geq 1.$
We have the results for Bob similarly as $\mathcal{S}_b^{t}(i_k)=p_b(1-p_b)^{k-1}$ when $ k\geq 1$.
\sjqxzlhy{Thus Eqn~(\ref{eqn:epsilon})} can be further written as $Pr[\phi_t=1] =\sum^n_{i=1}\mathcal{S}_a^{t}(i)\mathcal{S}_b^{t}(i)
 \leq  \sum^n_{k=1}\mathcal{S}_a^{t}(i_k)\mathcal{S}_b^{t}(j_k)
 \leq  \sum^n_{i=1}p_ap_b(1-p_a)^{i-1}(1-p_b)^{i-1}.$
\end{proof}

Based on Theorem~\ref{thm:perfC} and Theorem~\ref{thm:upperbound},
the following corollary is straightforward.
\begin{corollary}
In \sjqxzlhy{the} independent dynamic environment, no matter
\sjqxzlhy{whether} there is a common clock or not, Strategy C is an
optimal strategy which achieves the minimum
expected TTR.
\end{corollary}

\subsection{General Cases in Dynamic Environment}
\subsubsection{Model}
We model the channel availability in general dynamic environments as a Markov process. At any time slot $t\geq 2$, for Alice, a channel that is open at time $t-1$ will become closed with probability $a_0$, \sjqxzlhy{and} a channel that is closed at $t-1$ will be open with probability $a_1$. $b_0$ and $b_1$ are similarly defined for Bob. As we assume at each round a channel of Alice (Bob) has a probability of $p_a $ ($p_b$) to be open, it is easy to obtain $p_a(1-a_0)+(1-p_a)a_1 = p_a$ and $p_aa_0+(1-p_a)(1-a_1) =1-p_a $.
Therefore, $a_0$ and $a_1$ should satisfy $a_0=a_1=0$, or $\frac{a_0}{a_1}=\frac{1-p_a}{p_a}$ where $ 0< a_0,a_1 \leq 1$.
Here, we can define a parameter $\lambda_a$ and set $a_0=\lambda_a(1-p_a)$,  $a_1=\lambda_ap_a$.
where $0\leq \lambda_a \leq \min(\frac{1}{p_a},\frac{1}{1-p_a})\leq
2$.

We call $\lambda_a$ the \emph{environment dynamic factor}
\sjqxzlhy{of} Alice. Similarly $\lambda_b$ is defined for Bob. We have $0\leq \lambda_b \leq
\min(\frac{1}{p_b},\frac{1}{1-p_b})\leq 2$, \sjqxzlhy{and} $b_0=\lambda_b(1-p_b)$, $b_1=\lambda_bp_b.$
Now, we can see the parameters $p_a$ and $p_b$ reflect the channel open probabilities at a round, and $\lambda_a$ and $\lambda_b$ reflect the dynamic of channel availability over rounds. Moreover, the stable environment discussed in
Section~\ref{sec:stable} is actually a special case when $\lambda_a=\lambda_b=0$. The semi-stable environment is the case that $\lambda_a=\lambda_b=2$, and the independent dynamic environment is $\lambda_a=\lambda_b=1$. The closer to $1$ the dynamic factor, the more dynamic the environment. We regard the
environment dynamic factors as constants, which is reasonable in
real applications.

\vspace{-2mm}\subsubsection{Performance of Strategy C in General Cases}
Similar to the argument in Appendix~\ref{app:azar}, as Alice and Bob do not know the channel
status of each other, the expected TTR for a
dynamic environment also satisfies
$E[TTR]=\Omega(\frac{1}{\min(p_a,p_b)})$. The following theorem gives an
asymptotically matching upper bound of the TTR
by analyzing the performance of Strategy C in general dynamic
environments when neither of the environment at Alice and Bob is stable or semi-stable, i.e., $ \lambda_a, \lambda_b \not \in \{0,2\}$. Its proof is deferred to Appendix~\ref{app:thm5} due to space limitations.

\begin{theorem}\label{thm:generalcase}
When $n\rightarrow +\infty$, the expected TTR of Strategy C in the dynamic environments where $\lambda_a,\lambda_b \notin \{0,2\}$ is $O\left(\ln\left(\frac{1}{\min(p_a,p_b)}\right)\cdot \left(\frac{1}{p_a}+\frac{1}{p_b}-1\right)\right),$
which is optimal up to a logarithmic factor $\ln(\frac{1}{\min(p_a,p_b)})$.
\end{theorem}

%
%
\vspace{-4mm}\section{Simulation}\label{sec:sim}

In this section, we will carry out numerous simulations to
demonstrate the efficiency of our strategies in different system
settings, which validate our theoretical analysis. We also try to
exploit the impacts of different parameters on performance of the
strategies, such as the number of channels $n$, the channel open
possibilities $p_a$ and $p_b$, the environment dynamic factors $\lambda_a$ and $\lambda_b$, and the intruder.

Our simulation includes $4$ different strategies, Strategy A, B, C,
and random, which means Alice and Bob both choose an open channel
randomly at each round.  For simplicity, we set $\lambda_a=\lambda_b=\lambda$ and
$p_a=p_b=p$. For each set of the parameters $n$, $p$ and $\lambda$,
we run simulations \hstanq{ a large number of times} for each strategy,  and take the
average TTR as $E(TTR)$. To make sure there is at least one common channel in stable environments, we check each case generated and drop the ones of no common channels. For dynamic environments, we allow no common channels temporarily in some time slots.


During the simulation, according to the real number of channels in white spaces, we set $n\in [20,100]$ and mostly focus on $n\in\{20,30,50\}$. Moreover, we simulate more extensively the cases with a large $p$, i.e., $p\geq 0.6$, because 1) it is assumed that the channels in white spaces have a high probability to open for SUs; and 2) it will guarantee there are common channels with a high probability when $n$ is relatively small.
\subsection{Performance of the Strategies}

We first describe the performance of the $4$ different strategies. Table~\ref{table:exp_performace} shows the results with different
settings of $n$, $\lambda$ and $p$. Recall that Strategy C is not
applicable to stable environments ($\lambda=0$) (Refer to
Section~\ref{sec:dynamic:independent}).


\begin{table}[t]
\centering \caption{Performance of the
Strategies}\label{table:exp_performace}
\begin{tabular}{p{0.07\textwidth}|p{0.07\textwidth}|p{0.06\textwidth}|p{0.06\textwidth}|p{0.06\textwidth}||p{0.07\textwidth}|p{0.07\textwidth}|p{0.06\textwidth}|p{0.06\textwidth}|p{0.06\textwidth}}
\hline
 Setting $n=20$ & Strategy & $p = 0.6$ & $p = 0.75$ & $p = 0.9$ &Setting $n=50$ & Strategy & $p = 0.6$ & $p = 0.75$ & $p = 0.9$\\
 \hline \hline
 \multirow{3}{*}{$\lambda = 0$} & Random & 21.095 & 20.492 & 20.090 & \multirow{3}{*}{$\lambda = 0$} & Random & 50.227 & 50.430 & 50.518 \\
 \cline{2-5} \cline{ 7-10}
 & A & 39.877 & 21.490 & 14.087 & &A & 34.466 & 21.196 & 14.006 \\
 \cline{2-5} \cline{ 7-10}
 & B & \textbf{2.599} & \textbf{1.716} & \textbf{1.214} && B & \textbf{2.585} & \textbf{1.705} & \textbf{1.224} \\
 \hline
 \multirow{4}{*}{ $\lambda = 0.1$} & Random & 20.421 & 19.843 & 19.793 & \multirow{4}{*}{ $\lambda = 0.1$} & Random & 49.552 & 49.555 & 50.544\\
 \cline{2-5} \cline{ 7-10}
 & A & 35.566 & 20.715 & 13.750 && A & 32.206 & 20.627 & 14.158 \\
 \cline{2-5} \cline{ 7-10}
 & B & \textbf{2.534} & \textbf{1.697} & \textbf{1.228} && B & \textbf{2.531} & \textbf{1.694} & \textbf{1.222}\\
 \cline{2-5} \cline{ 7-10}
 & C & 9.111 & 5.564 & 2.999 && C & {8.956} & {5.808} & {2.864}  \\
 \hline
 \multirow{4}{*}{ $\lambda = 1$} & Random & 20.052 & 20.068 & 19.988 &\multirow{4}{*}{ $\lambda = 1$} & Random & 49.616 & 49.239 & 49.072 \\
 \cline{2-5} \cline{ 7-10}
 & A & 35.180 & 21.022 & 13.899 && A & 32.379 & 20.191 & 13.833  \\
 \cline{2-5} \cline{ 7-10}
 & B & 2.311 & 1.681 & 1.223&& B & 2.311 & 1.672 & 1.236 \\
  \cline{2-5} \cline{ 7-10}
 & C & \textbf{2.308} & \textbf{1.677} & \textbf{1.217}&& C & \textbf{2.299} & \textbf{1.658} & \textbf{1.222} \\
 \hline 
\end{tabular}
\begin{tablenotes}
\item [1] Remarks: As an example, the first entry $21.095$ means that, when $n=20$, $\lambda=0$ and $p=0.6$, the
expected TTR of the random strategy is $21.095$ rounds.
\end{tablenotes}
\end{table}


\vspace{-2mm}\begin{figure}[hbpt] 
\begin{minipage}[htbp]{0.48\textwidth}
\centering \epsfig{file=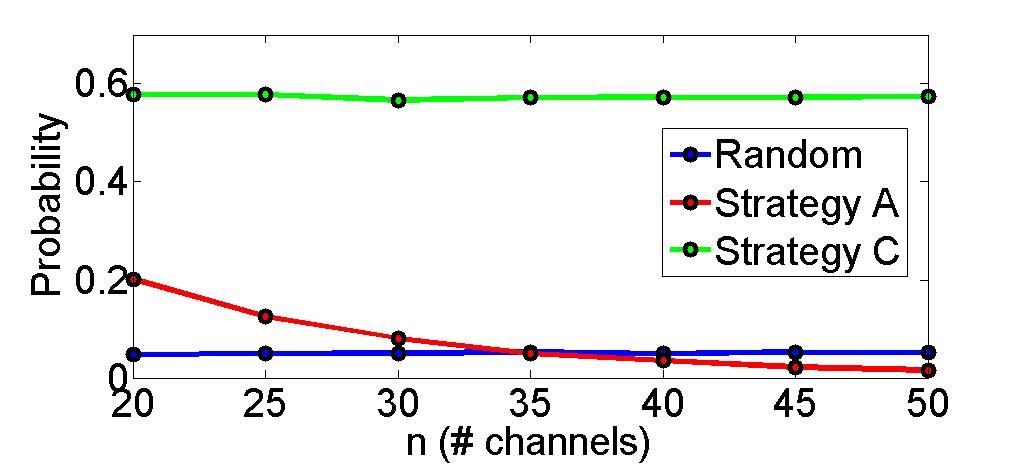, width=0.8\linewidth}
\end{minipage}
\begin{minipage}[htbp]{0.48\textwidth}
\caption{\label{fig:fail}The
probability of strategies with $E(TTR)\geq 3n$ in stable environments ($p = 0.6$).}\end{minipage}
\end{figure}\vspace{-2mm}

%


\begin{figure*}[t] 
\centering \subfigure[Strategy A
($\lambda=0$)]{\epsfig{file=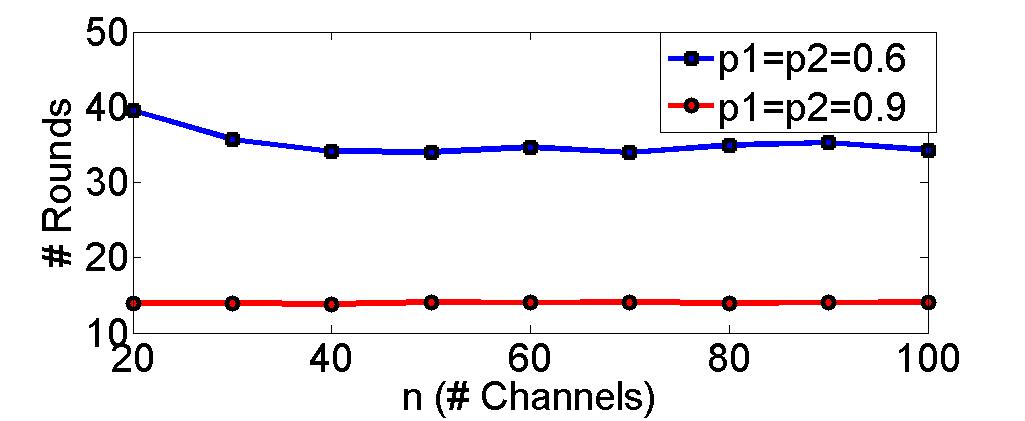, width=0.3\linewidth}}
\subfigure[Strategy B ($\lambda=0$)]{\epsfig{file=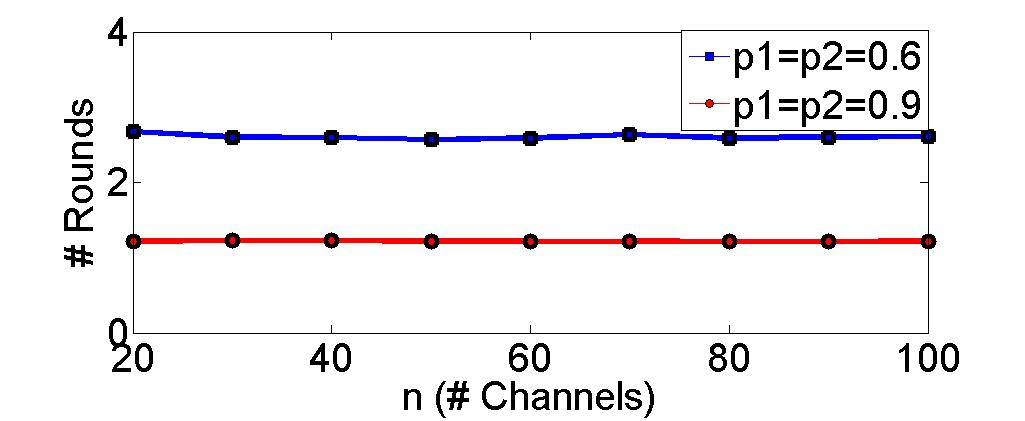,
width=0.3\linewidth}} \subfigure[Strategy C
($\lambda=1$)]{\epsfig{file=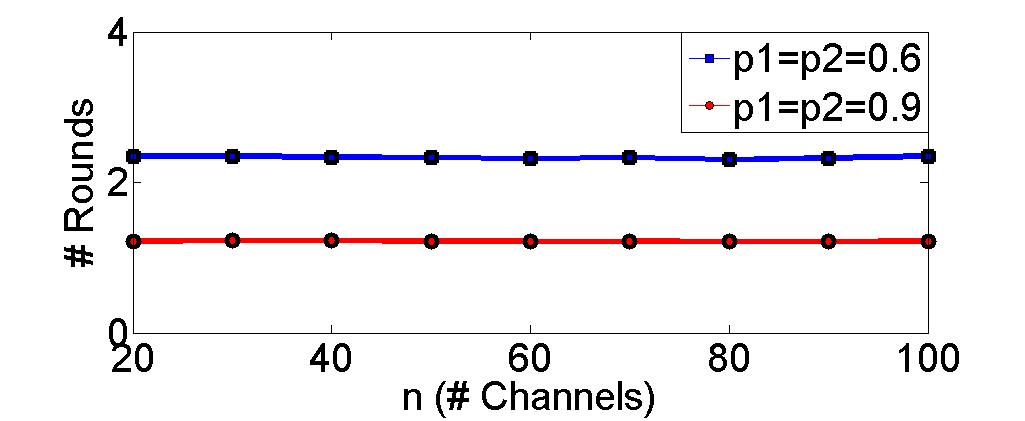, width=0.3\linewidth}}
\vspace{-2mm}\caption{\label{fig:exp_n}Performance of strategies with different
numbers of channels.}
\end{figure*}


\begin{figure*}[t] 
\centering \subfigure[Strategy A
($\lambda=0$)]{\epsfig{file=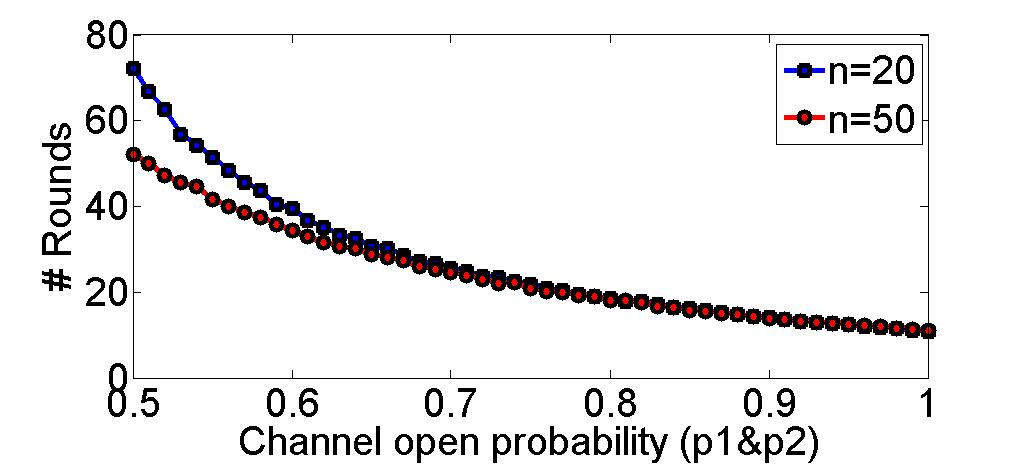, width=0.3\linewidth}}\label{fig:exp_p:A}
\subfigure[Strategy B ($\lambda=0$)]{\epsfig{file=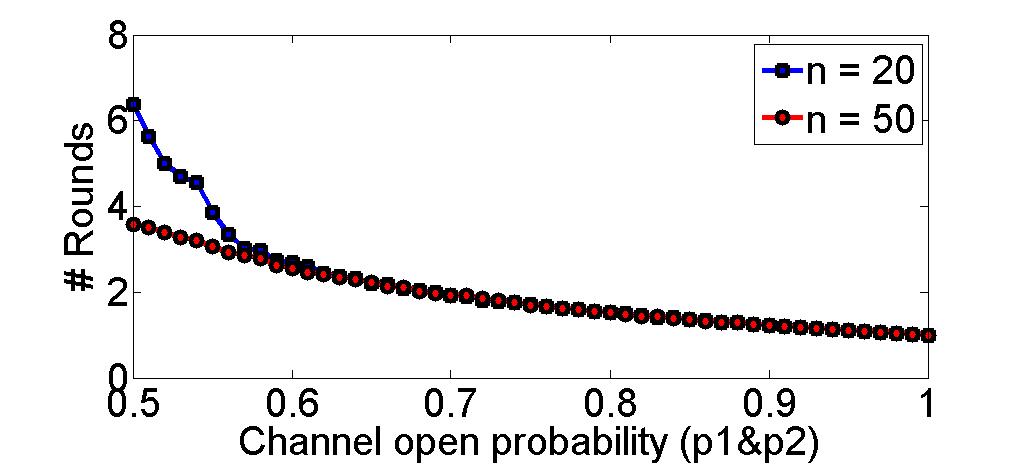,
width=0.3\linewidth}} \label{fig:exp_p:B}
\subfigure[Strategy C
($\lambda=1$)]{\epsfig{file=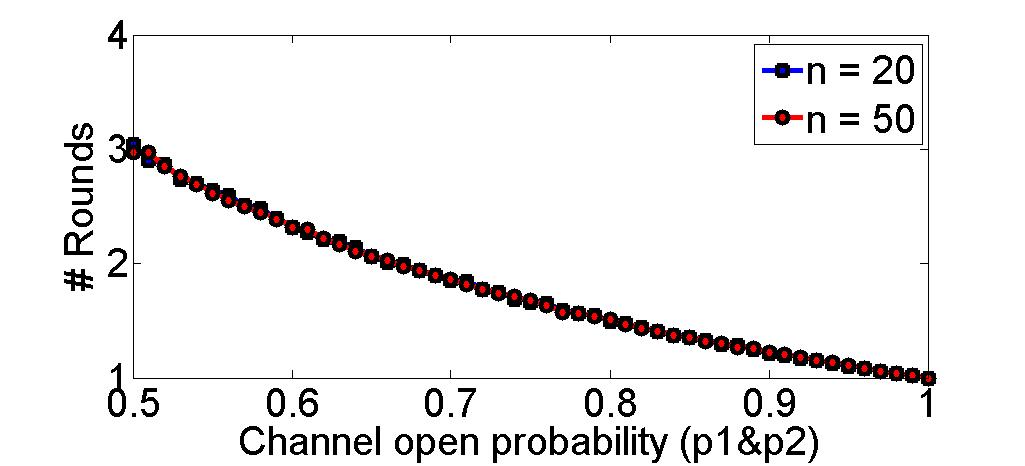, width=0.3\linewidth}}\label{fig:exp_p:C}
\vspace{-2mm}\caption{\label{fig:exp_p}Performance of strategies with different
channel open probabilities.}
\end{figure*}\vspace{-2mm}

We can find that our Strategy B can achieve a fast rendezvous and significantly outperforms the random strategy (Random) and Strategy A in stable environments ($\lambda=0$). However,
recall that Strategy B needs a common clock. In independent dynamic
environments ($\lambda=1$), Strategy C is optimal which validates
our theoretical analysis.

Another fact is that, although Strategy A is proved to be  optimal in
asynchronous stable environments~\cite{esaAzar}, it performs poorly in
dynamic environments compared with C which does not need a common
clock either, even when $\lambda$ is as small as $0.1$.

\hstanq{In addition, in stable environments, Strategy A may have a bad performance, worse than the random strategy, when $n$ is small (e.g., $n=20$). It is due to the property of the geometric distribution used in Strategy A.  We study the the probabilities of $E(TTR)\geq 3n$
for Strategies A, C and Random\footnote{We do not investigate B here as it guarantees a fast rendezvous with no more than $n$ rounds in stable synchronous environments.}. Figure \ref{fig:fail} shows the
results, which validates that Strategy A can not
handle cases well when there are only a small number of open channels at Alice and Bob. Moreover, Strategy C always has a high failure probability, which coincides with the fact that Strategy C is not applicable to stable environments.}

\subsection{Impacts of different parameters}

Now, we perform analyses to examine the impacts of different
parameters: (1) the channel number $n$; (2) the channel open
probability $p$; (3) the environment dynamic factor $\lambda$; and
(4) a third party: an intruder.

\subsubsection{The channel number $n$ }
We first examine how the channel number $n$ influences the
performance of the strategies. Figure \ref{fig:exp_n} shows the
results with different settings of $n$ when
$p_a = p_b=p \in \{0.6, 0.9\}$. We can see with relatively large
channel open possibilities, the performance is insensitive to $n$.
Only Strategy A \hstan{becomes worse} when $n \leq 30$,
which is caused by the properties of geometry distribution. In
addition, for a larger $p$, the influence of $n$ is smaller. It is
reasonable because with a large $p$ there has been plenty of common
channels even with a small $n$.

\subsubsection{The channel open probabilities $p_a$ and $p_b$ }
Figure \ref{fig:exp_p} shows the performance of strategies $A$, $B$
and $C$ with different channel open probabilities. We can see that with the increasing of the open
probability, the expected TTR for Strategies A, B and C
have a clear decreasing trend, which coincides with our theoretical
results stating that the expected time is a reciprocal function of the channel open probability. 

\vspace{-2mm}\subsubsection{The environment dynamic factor }
Table \ref{table:exp_performace} has already shown that Strategy C
achieves the best performance in independent dynamic environments
($\lambda=1$). Here, we quantitatively examine how the environment
dynamic factor $\lambda$ will affect Strategy C. Figure
\ref{fig:exp_a} illustrates the performance of Strategy C when
$\lambda \in [0.05, 1.2]$\footnote{Recall that for $p_a=p_b=p$,
$0\leq\lambda\leq \frac{1}{\min(p,1-p)}$
(Section~\ref{sec:dynamic}).}. We can see the more dynamic the
environment is (the closer $\lambda$ approaches to $1$), the better
Strategy C performs. It is also straightforward that with larger
$p$, the influence of $\lambda$ is smaller as there are plenty of
open channels.
\vspace{-3mm}\begin{figure}[hbpt] 
\begin{minipage}[htbp]{0.48\textwidth}
\centering \epsfig{file=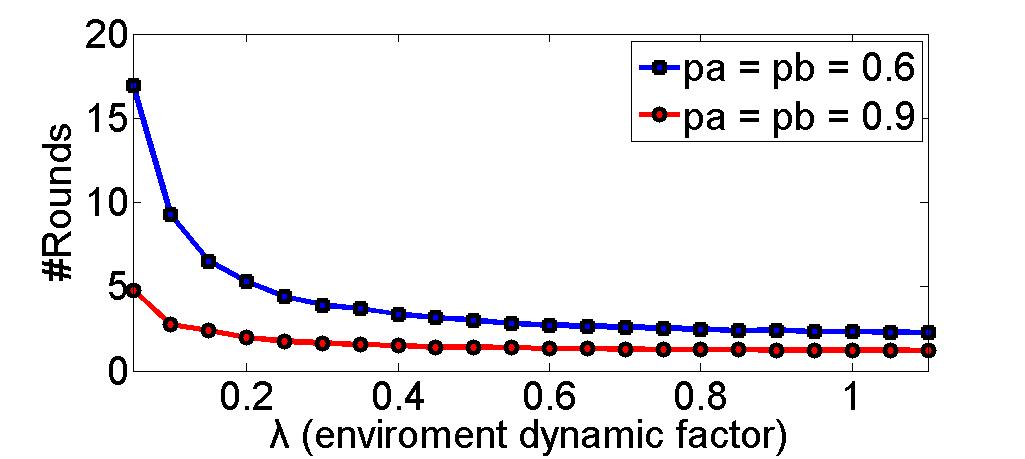, width=0.8\linewidth}
\end{minipage}
\begin{minipage}[htbp]{0.48\textwidth}
\caption{\label{fig:exp_a}Performance of Strategy C with different
environment dynamic factors ($n=30$).}\end{minipage}
\end{figure}\vspace{-6mm}

\subsubsection{A third party: an intruder }\label{sec:sim:impact:q}
A third party, i.e., an intruder, may block some channels between
Alice and Bob. \hstanq{ We assume an intruder opens each
channel with probability of $q$ and does not change the blocked channels over time.} One user does not know a closed channel at his side is due to the existence of the intruder or PUs.  Therefore, for the theoretical analysis of rendezvous strategies, this case is equivalent to the environment
where there are only Alice and Bob with channel open probabilities
of $p_aq$ and $p_bq$, respectively.

Here, we will examine the influence of an intruder in experiments, and Figure \ref{fig:q} illustrates the results.
We can find that 1) The larger the
channel open probability is, the smaller influence the intruder
causes, which is straightforward; and 2) the existence of an intruder
has the largest influence to Strategy A, which is due to the
property of geometry distribution used in A.
%

\begin{figure*}[htbp] 
\centering \subfigure[Strategy A
($\lambda=0$)]{\epsfig{file=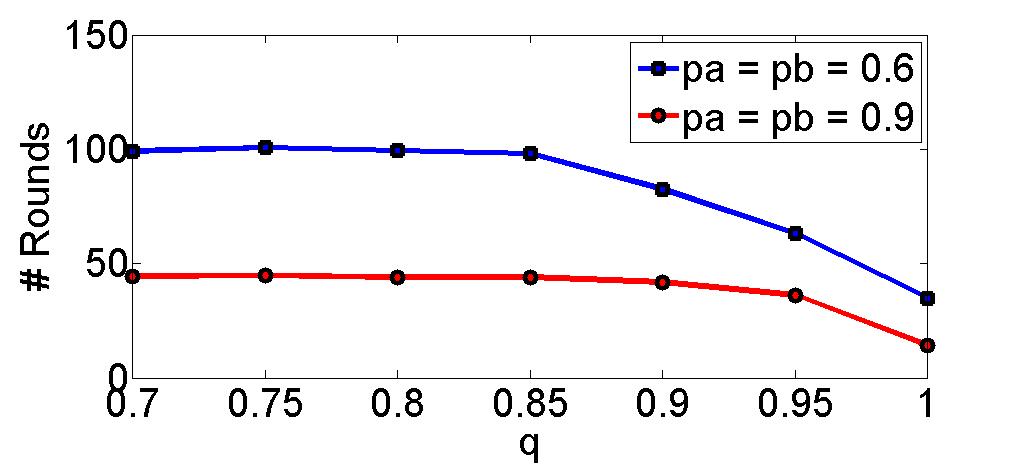, width=0.3\linewidth}}
\subfigure[Strategy B ($\lambda=0$)]{\epsfig{file=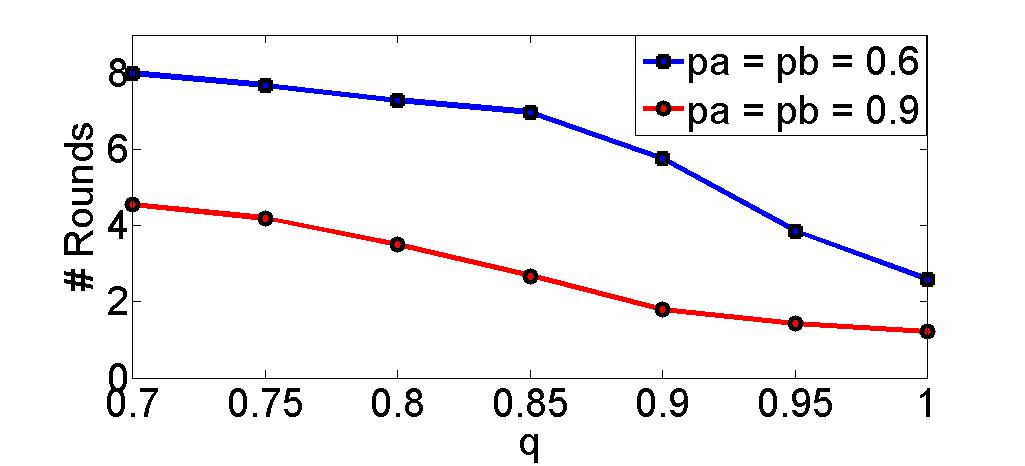,
width=0.3\linewidth}} \subfigure[Strategy C
($\lambda=1$)]{\epsfig{file=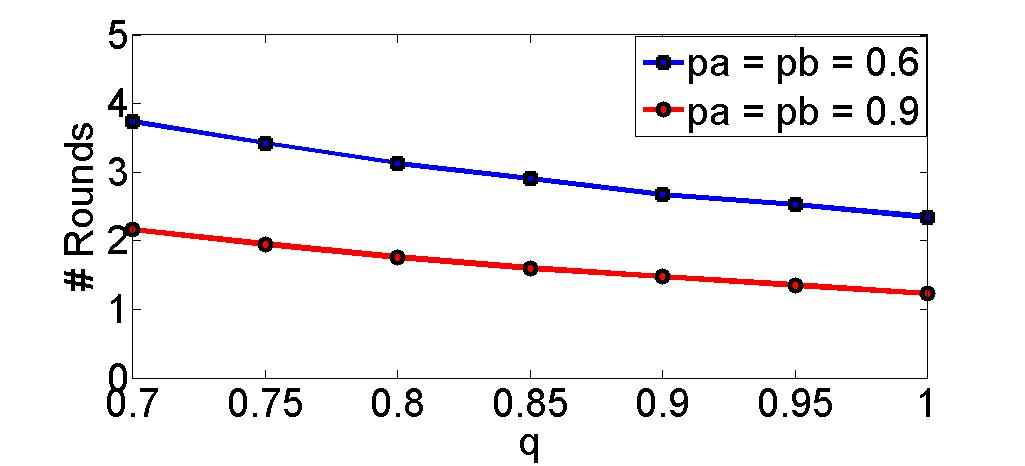, width=0.3\linewidth}}
\vspace{-2mm}\caption{\label{fig:q}Performance of strategies with existence of an
intruder (n=30).}
\end{figure*} \vspace{-6mm}

\section{Conclusion}\label{sec:con}
In this paper we study the rendezvous problem in cognitive radio
networks. For different system settings, such as asynchronous or
synchronous systems in stable or dynamic
environments, we derive various strategies and prove their
optimality in time-to-rendezvous. Simulations have been carried out to
demonstrate the efficiency of our strategies.
The impacts of different parameters on the
TTR are also investigated. In our current work, each
secondary user can only access one channel at a time slot. Designing
optimal rendezvous strategies for SUs that can access multiple
(continuous) channels at a time is an interesting extension of this
work. In addition, to achieve rendezvous among more than two nodes,
where there are additional challenges such as the interference between
simultaneously transmitting nodes, is also an exciting
direction of the future work.

\newpage


\begin{appendices}
\titleformat{\section}{\Large \bfseries}{Appendix \thesection}{1em}{}
\section{Azar et al.'s work} \label{app:azar}
For stable asynchronous environments, Azar et al. proposed in \cite{esaAzar} a stationary strategy based on
geometric distributions shown as Strategy A:

\begin{table}[htbp]
\centering
\begin{tabular}{p{0.8\textwidth}}
 \hline
Strategy A:\\
\hline At each round,\\
 \quad $\mathcal {S}_a$:  Alice chooses her
$i$-th open channel ($i\geq 1$) with possibility of
$\frac{p_b}{6} (1-\frac{p_b}{6})^{i-1}$. \\
\quad $\mathcal {S}_b$:  Bob chooses his $i$-th open channel ($i\geq
1$) with possibility of
$\frac{p_a}{6}(1-\frac{p_a}{6})^{i-1}$. \\
\hline
\end{tabular}
\end{table}
They also claimed \hstan{a lower bound of the expected TTR as $\frac{1}{\min(p_a,p_b)}$}, which holds for any strategies in
the stable environments. The proof is trivial.  Alice is oblivious of Bob's
channel status. No matter what strategy she uses, at a round, when Alice
chooses
a channel $i$, the possibility that channel $i$ is open at Bob is $p_b$.
Therefore, on the side of Alice, the rendezvous possibility at a round is
no more than $p_b$. A similar result holds for Bob. Then, at each round, the
rendezvous probability can not exceed $\min(p_a,p_b)$. \hstan{Therefore, the lower bound is achieved.}

\section{Proof of Theorem~\ref{thm:generalcase}}\label{app:thm5}
\begin{proof}
We first explain why the argument for the independent dynamic environment is not applicable to the general cases.
Take Alice as an example. Given that channel $i$ of Alice is open at
time $t$, it will have a probability of $a_0=\lambda_a(1-p_a)$ to
change to close at time $t+1$. Hence, the open probability of
channel $i$ at time $t+1$ is $1-\lambda_a(1-p_a)$, which in general may not be equal to
$p_a$. Consequently, the events of a channel being open in different time slots are not independent, which makes the previous argument fail when applying to the general cases.

In order to analyze the performance of Strategy C, our idea
is to find some special time slots such that, restricted on these
slots alone, the environment becomes close to the independent dynamic environment for which a tight upper bound on the TTR has already been obtained. We then argue that such ``closeness'' can guarantee a similar time upper bound, which gives the desired result.

We now formalize the above idea. Let $R$ be a parameter to be specified later. We restrict Strategy C on the $(jR+1)$-th rounds for all non-negative integers $j$, or equivalently, consider a new strategy $C'$ which is identical to $C$ on the $(jR+1)$-th rounds, but in other rounds both Alice and Bob do nothing, i.e., not trying to connect with each other. Obviously the TTR of $C'$ is at least as large as that of Strategy $C$. Hence, a proper upper bound on the TTR of $C'$ will suffice for our purpose.

For $i\in\{1,2,\ldots,n\}$ and $K,M\geq 1$, let $\mathcal{P}_i^{K+M}$ denote the probability that channel $i$ is open for Alice at round $K+M$ given $A_i^K$, i.e., the availability of channel $i$ for Alice at round $K$. Let $\mathcal{P}_i^K:=A_i^K$. (Rigorously speaking $A_i^k$ is a random variable which takes value 0 or 1. Nonetheless, as shown in the following, the actual value of $A_i^K$ does not affect the result. Thus we treat $A_i^K$ as a constant.) Recalling that $a_0=\lambda_a(1-p_a)$ and $a_1=\lambda_a p_a$, it is clear that for any integer $M\geq 0$,
\begin{eqnarray*}
\mathcal{P}_i^{K+M+1}&=&(1-\lambda_a(1-p_a))\mathcal{P}_i^{K+M}+\lambda_a p_a(1-\mathcal{P}_i^{K+M})\\
&=&(1-\lambda_a)\mathcal{P}_i^{K+M}+\lambda_a p_a.
\end{eqnarray*}
Rearranging terms gives that for any $M\geq 0$,
\begin{eqnarray*}
\mathcal{P}_i^{K+M+1}-p_a = (1-\lambda_a)(\mathcal{P}_i^{K+M}-p_a),
\end{eqnarray*}
from which it follows that
\begin{equation}\label{eqn:general}
\hstan{\mathcal{P}_i^{K+R}-p_a = (1-\lambda_a)^R(\mathcal{P}_i^{K}-p_a)=(1-\lambda_a)^R(A_i^{K}-p_a).}
\end{equation}
Noting that $(A_i^{K}-p_a)$ is a constant independent of $R$ and that $|1-\lambda_a|<1$, we have
\begin{eqnarray*}
\mathcal{P}_i^{K+R}-p_a \rightarrow 0 \textrm{~~as~} R \rightarrow +\infty .
\end{eqnarray*}
Therefore, for any $j\geq 1$, the open probability for channel $i$ at round $K+jR$ given its status at round $K$ can be arbitrarily close to $p_a$, provided that $R$ is sufficiently large (which can always be guaranteed when $n\rightarrow +\infty$).
So the events that channel $i$ is open for Alice at round $jR+1$, for all $j\geq 0$, are ``approximately'' independent from each other. More precisely, by choosing
$$R=\max\left\{\frac{\ln (1/(0.001p_a))}{\ln (1/|1-\lambda_a|)}, \frac{\ln (1/(0.001p_b))}{\ln (1/|1-\lambda_b|)} \right\},$$
we can obtain from Eqn~(\ref{eqn:general}) that $|\mathcal{P}_i^{K+R}-p_a| \leq 0.001p_a$, which implies that
$0.999p_a \leq \mathcal{P}_i^{K+R} \leq 1.001p_a$. (Here the constant 0.001 is just an illustration; it can be arbitrarily small.)
Similar results also hold for Bob.

Then, in the $(jR+1)$-th round for each $j\geq 1$, the probability
that Alice and Bob both find channel $i$ is at least
$(0.999/1.001)^2\sum_{i=1}^n
(1.001p_a(1-1.001p_a)^{i-1})\times(1.001p_b(1-1.001p_b)^{i-1})$, which can
be regarded as $(0.999/1.001)^2$ times the probability in
Eqn~(\ref{eqn:14}) with $p_a$ and $p_b$ replaced with $1.001p_a$ and
$1.001p_b$ respectively. Then, by routine probability calculations
similar to Theorem~\ref{thm:perfC}, the expected TTR is
at most
$\left(\frac{1.001}{0.999}\right)^2 \times \left(\frac{1}{1.001p_a}+\frac{1}{1.001p_b}-1\right)=O\left(\frac{1}{p_a}+\frac{1}{p_b}-1\right)$.
Finally note that, since only the $(jR+1)$-th rounds are considered,
the actually upper bound on the expected TTR should be
$R$ times the previous bound. By our choice we have
$R=O\left(\ln(1/\min(p_a,p_b))\right)$, as the
environment dynamic factors are regarded as constants.
\end{proof}
\end{appendices}


\end{document}